\documentclass[12pt]{scrartcl}

\usepackage{amsthm}
\usepackage{amsmath}
\usepackage{amssymb}
\usepackage{bbm}
\usepackage{hyperref}
\usepackage{algorithm}
\usepackage[noend]{algorithmic}
\usepackage{graphicx}
\usepackage{lineno}
\usepackage{stmaryrd}

\newcommand{\tit}[1]{\text{\itshape{#1}}}
\newcommand{\fshift}{\tit{shift}\,} 
\newcommand{\fcost}{\xi}   
\newcommand{\diffdaa}[2]{\tit{DiffDAA}(#1,#2)}
\newcommand{\shift}[2]{\fshift^{#1}({#2})}        
\newcommand{\cost}[2]{\fcost^{#1}({#2})}             
\newcommand{\rightpos}[2]{\tit{right}\,^{#1}({#2})} 
\newcommand{\ashift}{\texttt{shift}} 
\newcommand{\daaval}{\tit{value}}

\newcommand{\rev}[1]{\stackrel{\leftarrow}{{#1}}}
\newcommand{\suff}[1]{{#1}'} 

\newcommand{\setN}{\mathbbmss{N}}
\newcommand{\setZ}{\mathbbmss{Z}}
\newcommand{\Prob}{\mathbbmss{P}}
\newcommand{\prob}{\Prob}

\newcommand{\len}[1]{{\vert {#1} \vert}}
\newcommand{\chr}[2]{#1[#2]}

\newcommand{\prefix}[2]{#1[..#2]}
\newcommand{\substr}[3]{{#1}[{#2}\ldots{#3}]}

\newcommand{\dist}{\mathcal{L}}

\newcommand{\paa}{\mathcal{P}}
\newcommand{\daa}{\mathcal{D}}

\newcommand{\emptystring}{\varepsilon}
\newcommand{\iverl}{\llbracket}
\newcommand{\iverr}{\rrbracket}

\newcommand{\stateset}{\mathcal{Q}}
\newcommand{\emiset}{\mathcal{E}}
\newcommand{\valset}{\mathcal{V}}
\newcommand{\valsize}{\vartheta}
\newcommand{\stateproc}{Q}
\newcommand{\emiproc}{E}
\newcommand{\valproc}{V}
\newcommand{\state}{q}
\newcommand{\emi}{e}
\newcommand{\val}{v}
\newcommand{\op}{\theta}
\newcommand{\emidist}{\mu}
\newcommand{\paatrans}{T}
\newcommand{\daaemi}{\eta}
\newcommand{\totalsteps}{n}

\newcommand{\paatuple}{\big(\stateset,\state_0,\paatrans,\valset,\val_0, \emiset, \emidist=(\emidist_\state)_{\state\in\stateset},\op=(\op_\state)_{\state\in\stateset}\big)}

\theoremstyle{definition}
\newtheorem{definition}{Definition}

\newtheorem{lemma}{Lemma}
\newtheorem{theorem}{Theorem}

\title{Exact Analysis of Pattern Matching Algorithms with Probabilistic Arithmetic Automata}
\author{Tobias Marschall and Sven Rahmann\\
\small Bioinformatics for High-Throughput Technologies,\\[-0.8ex]
\small Computer Science XI, TU Dortmund, Germany\\[-0.8ex]
\small \texttt{firstname.lastname@tu-dortmund.de}
}
\date{}

\begin{document}
\maketitle
\begin{abstract}
\noindent
We propose a framework for the exact probabilistic analysis of window-based pattern matching algorithms, such as Boyer-Moore, Horspool, Backward DAWG Matching, Backward Oracle Matching, and more.
In particular, we show how to efficiently obtain the distribution of such an algorithm's running time cost for any given pattern in a random text model, which can be quite general, from simple uniform models to higher-order Markov models or hidden Markov models (HMMs). 
Furthermore, we provide a technique to compute the exact distribution of \emph{differences} in running time cost of two algorithms.
In contrast to previous work, our approach is neither limited to simple text models, nor to asymptotic statements, nor to moment computations such as expectation and variance.
Methodically, we use extensions of finite automata which we call \emph{deterministic arithmetic automata} (DAAs) and \emph{probabilistic arithmetic automata} (PAAs)~\cite{Marschall2008}.
To our knowledge, this is the first time that substring- or suffix-based pattern matching algorithms are analyzed exactly.
Experimentally, we compare Horspool's algorithm, Backward DAWG Matching, and Backward Oracle Matching on prototypical patterns of short length and provide statistics on the size of minimal DAAs for these computations.
\end{abstract}


\section{Introduction}
\label{sec:intro}

The basic pattern matching problem is to find all occurrences of a \emph{pattern} string in a (long) \emph{text} string, with few character accesses.
Let~$\ell$ be the text length and~$m$ be the pattern length.
The well-known Knuth-Morris-Pratt algorithm~\cite{KnuthMorrisPratt1977} reads each text character exactly once from left to right and hence needs exactly $\ell$ character accesses for any text of length $\ell$, after preprocessing the pattern in~$\Theta(m)$ time.
In contrast, the Boyer-Moore~\cite{Boyer1977}, Horspool~\cite{Horspool1980}, Sunday~\cite{Sunday1990}, Backward DAWG Matching (BDM,~\cite{Crochemore1994}) and Backward Oracle Matching (BOM,~\cite{Allauzen2001}) algorithms move a length-$m$ search window across the text and first compare its \emph{last} character to the last character of the pattern.
This often allows to move the search window by more than one position (at best, by~$m$ positions if the last window character does not occur in the pattern at all), for a best case of $n/m$, but a worst case of $mn$ character accesses. 
The worst case can often be improved to~$\Theta(m+n)$, but this makes the code more complicated and seldom provides a speed-up in practice.
For practical pattern matching applications, the most important algorithms are
Horspool, BDM (often implemented as Backward Nondeterministic DAWG Matching, BNDM, via a non-deterministic automaton that is simulated in a bit-parallel fashion), and BOM, depending on alphabet size, text length and pattern length; see~\cite[Section~2.5]{Navarro2002} for an experimental map.

A question that has apparently so far not been investigated is about the exact probability distribution of the number of required character accesses~$X^p_\ell$ when matching a given pattern~$p$ against a random text of finite length~$\ell$ (non-asymptotic case), even though related questions have been answered in the literature.
For example, \cite{Baeza-Yates1990,BaezaYatesRegnier1992} analyze the expected value of~$X^p_\ell$ for the Horspool algorithm.
In~\cite{Mahmoud1997} it is further shown that for the Horspool algorithm, $X^p_\ell$ is asymptotically normally distributed for i.i.d.\ texts, and~\cite{Smythe2001} extends this result to Markovian text models. In \cite{Tsai2006}, a method to compute mean and variance of these distributions is given.

In contrast to these results that are special to the Horspool algorithm, we use a general framework called \emph{probabilistic arithmetic automata} (PAAs), introduced at CPM'08~\cite{Marschall2008}, to compute the exact distribution of~$X^p_\ell$ for any window-based pattern matching algorithm. In~\cite{Marschall2008}, PAAs where introduced in order to compute the distribution of occurrence counts of patterns; the fact that they can also be used to analyze pattern matching algorithms further highlights their utility.
In our framework, the random text model can be quite general, from simple i.i.d.\ uniform models to high-order Markov models or HMMs.
The approach is applied exemplarily to the following pattern matching algorithms in the non-asymptotic regime (short patterns, medium-length texts): Horspool, B(N)DM, BOM. 
We do not treat BDM and BNDM separately as, in terms of text character accesses, they are indistinguishable (see Section~\ref{sec:bdm}).

This paper is organized as follows.
In the next section, we introduce notation and give a brief review of the Horspool, B(N)DM and BOM algorithms.
In Section~\ref{sec:daa}, we define \emph{deterministic arithmetic automata} (DAAs).
In Section~\ref{sec:construction}, we present a simple general DAA construction for the analysis of window-based pattern matching algorithms. 
We also show that the state space of the DAA can be considerably reduced by adapting DFA minimization to DAAs.
In Section~\ref{sec:paa}, we summarize the PAA framework with its generic algorithms, define finite-memory text models and connect DAAs to PAAs.
This yields, for each given pattern~$p$, algorithm, and random text model, a PAA that computes the distribution of~$X^p_\ell$ for any finite text length $\ell$.
Section~\ref{sec:comparison} introduces \emph{difference DAAs} by a product construction that allows to compare two algorithms on a given pattern.
Exemplary results on the comparison of several algorithms can be found in Section~\ref{sec:case_studies}. 
There, we also provide statistics on automata sizes for different algorithms and pattern lengths. 
Section~\ref{sec:discussion} contains a concluding discussion.

An extended abstract of this work has been presented at LATA'10~\cite{Marschall2010} with two alternative DAA constructions. 
In contrast to that version, the DAA construction in the present paper can be seen as a combination of both of those, and is much simpler.
Additionally, the DAA minimization introduced in the present paper allows the analysis of much longer patterns in practice. While~\cite{Marschall2010} was focused on Horspool's and Sunday's algorithms, here, we give a general construction scheme applicable to any window-based pattern matching algorithm and discuss the most relevant algorithms, namely Horspool, BOM, and B(N)DM, as examples.


\section{Algorithms}
\label{sec:algo}

Both pattern and text are over a finite alphabet~$\Sigma$.
Indexing generally starts at zero.
The pattern~$p=\chr{p}{0}\dots \chr{p}{m-1}$ is of length~$m$;
a (concrete) text~$s$ is of length~$\ell$.
By $\rev{p}$, we denote the reverse pattern $\chr{p}{m-1}\dots \chr{p}{0}$.

In the following, we summarize the Horspool, B(N)DM and BOM algorithms;
algorithmic details can be found in~\cite[Chapter 2]{Navarro2002}.

We do not discuss the Knuth-Morris-Pratt algorithm because its number of text character accesses is constant: Each character of the text is looked at exactly once.
Therefore, $\dist(X^p_\ell)$ is the Dirac distribution on $\ell$, i.e., $\Prob(X^p_\ell=\ell)=1$.

We also do not discuss the Boyer-Moore algorithm, since it is never the best one in practice because of its complicated code to achieve optimal asymptotic running time.
In contrast to our earlier paper~\cite{Marschall2010}, we also skip the Sunday algorithm, as it is almost always inferior to Horspool's. Instead, we focus on those algorithms that are fastest in practice according to \cite[Fig.~2.22]{Navarro2002}.

The Horspool, B(N)DM and BOM algorithms have the following properties in common:
They maintain a search window~$w$ of length~$m=\len{p}$ that initially starts at position~$0$ in the text~$s$, such that its rightmost character is at position~$t=m-1$. 
The right window position~$t$ grows in the course of the algorithm; we always have~$w = \substr{s}{(t-m+1)}{t}$.
The algorithms look at the characters in each window from right to left, and thus compare the reversed window with the reversed pattern $\rev{p}$.
(For Horspool, variants with different comparison orders are possible, but the rightmost character is always compared first.)

The two properties of an algorithm that influence our analysis are the following:
For a pattern $p\in\Sigma^m$, each window $w\in\Sigma^m$ determines
\begin{enumerate}
\item its cost $\cost{p}{w}$, e.g., the number of text character accesses required to analyze this window,
\item its shift $\shift{p}{w}$, which is the number of characters the window is advanced after it has been examined.
\end{enumerate}

\subsection{Horspool}
First, the rightmost characters of window and pattern are compared; that means, $a := \chr{w}{m-1}=\chr{s}{t}$ is compared with~$\chr{p}{m-1}$. 
If they match, the remaining~$m-1$ characters are compared until either the first mismatch is found or an entire match has been verified.
This comparison can happen right-to-left, left-to-right, or in an arbitrary order that may depend on~$p$. 
In our analysis, we focus on the right-to-left case for concreteness, but the modifications for the other cases are straightforward.
Therefore, the cost of window~$w$ is
\[
   \cost{p}{w} = \begin{cases}
   m &\text{ if } p=w,\\
   \min \{i: 1\leq i\leq m,\; \chr{p}{m-i} \neq \chr{w}{m-i} \} &\text{ otherwise}.
   \end{cases}
\]

In any case, the rightmost window character~$a$ is used to determine how far the window can be shifted for the next iteration.
The shift-function ensures that no match can be missed by moving the window such that~$a$ becomes aligned to the rightmost~$a$ in~$p$ (not considering the last position). 
If~$a$ does not occur in~$p$ (or only at the last position), it is safe to shift by~$m$ positions. Formally, we define
\begin{align*}
   \rightpos{p}{a} &:= \max\big[ \{i\in\{0,\dots,m-2\}: \chr{p}{i}=a\} \cup \{-1\} \big]\,,\\
   \ashift[a]      &:= (m-1) - \rightpos{p}{a} \,, \text{ assuming $p$ fixed,}\\
   \shift{p}{w}    &:= \ashift[\chr{w}{m-1}] \,.
\end{align*}

For concreteness, we state Horspool's algorithm and how we count text character accesses as pseudocode in Algorithm~\ref{alg:horspool}. 
Note that after a shift, even when we know that~$a$ now matches its corresponding pattern character, the corresponding position is compared again and counts as a text access. 
Otherwise the additional bookkeeping would make the algorithm more complicated; this is not worth the effort in practice. 
The lookup in the~$\ashift$-table does not count as an additional access, since we can remember~$\ashift[a]$ as soon as the last window character has been read.

\begin{algorithm}[t!]
\caption{\textsc{Horspool-with-Cost}}\label{alg:horspool}
\begin{algorithmic}[1]
\vspace*{.1cm}
\REQUIRE text $s\in\Sigma^\ast$, pattern $p\in\Sigma^m$
\ENSURE pair (number $occ$ of occurrences of $p$ in $s$, number $cost$ of accesses to $s$)
  \STATE pre-compute table $\ashift[a]$ for all $a\in\Sigma$ 
	\STATE $(occ,cost) \gets (0,0)$
 	\STATE $t \gets m-1$
	\WHILE{$t < \len{s}$}\label{alg:hor_mainloop_begin}
		\STATE $i\gets 0$
		\WHILE{$i<m$}\label{alg:hor_innerloop_begin}
			\STATE $cost\gets cost+1$
			\STATE \textbf{if} $\chr{s}{t-i}\neq\chr{p}{(m-1)-i}$ \textbf{then break}
			\STATE $i\gets i+1$
		\ENDWHILE\label{alg:hor_innerloop_end}
		\STATE \textbf{if} $i=m$ \textbf{then} $occ\gets occ+1$
		\STATE $t\gets t+\ashift[\chr{s}{t}]$
	\ENDWHILE\label{alg:hor_mainloop_end}
	\RETURN $(occ, cost)$\label{alg:hor_n_alteration}
\end{algorithmic}
\end{algorithm}

The main advantage of the Horspool algorithm is its simplicity.
Especially, a window's shift value depends only on its last character,
and its cost is easily computed from the number of consecutive matching characters at its right end.
The Horspool algorithm does not require any advanced data structure and can be implemented in a few lines of code.

\subsection{Backward (Nondeterministic) DAWG Matching, B(N)DM}\label{sec:bdm}
The main idea of the BDM algorithm is to build a deterministic finite automaton (in this case, a suffix automaton, which is a directed acyclic word graph or DAWG) that recognizes all substrings of the reversed pattern, accepts all suffixes of the reversed pattern (including the empty suffix), and enters a FAIL state if a string has been read that is not a substring of the reversed pattern.

The suffix automaton processes the window right-to-left.
As long as the FAIL state has not been reached, we have read a substring of the reversed pattern.
If we are in an accepting state, we have even found a suffix of the reversed pattern (i.e., a prefix of~$p$). Whenever this happens before we have read $m$ characters, the last such event marks the next potential window start that may contain a match with $p$, and hence determines the shift.
When we are in an accepting state after reading $m$ characters, we have found a match, but this does not influence the shift.

So, $\cost{p}{w}$ is the number of characters read when entering FAIL,
or $m$ if $p=w$.
Let $I^p(w)\subseteq\{0,\ldots,m-1\}$ be the set defined by $i\in I^p(w)$ if and only if the suffix automaton of $\rev{p}$ is in an accepting state after reading $i$~characters of~$w$. 
Then
\begin{align*}
  \shift{p}{w} &= \min \big\{ m-i \,\big|\, i\in I^{p}(w)\big\}.
\end{align*}
Note that $I^{p}(w)$ is never empty as the suffix automaton accepts the empty string and, thus, $0\in I^{p}(w)$ for all windows $w$.

The advantage of BDM are long shifts, but its main disadvantage is the necessary construction of the suffix automaton, which is possible in $O(m)$ time via the suffix tree of the reversed pattern, but too expensive in practice to compete with other algorithms unless the search text is extremely long.

Constructing a nondeterministic finite automaton (NFA) instead of the deterministic suffix automaton is much simpler. However, processing a text character then does not take constant, but $O(m)$ time.
However, the NFA can be efficiently simulated with bit-parallel operations such that processing a text character takes $O(m/W)$ time, where $W$ is the machine word size.
For many patterns in practice, this is as good as $O(1)$.
The resulting algorithm is then called BNDM.

From the ``text character accesses'' analysis point of view, BDM and BNDM are equivalent, as they have the same shift and cost functions.

\subsection{Backward Oracle Matching, BOM}
BOM is similar to BDM, but the suffix automaton of the reversed pattern is replaced by a simpler deterministic automaton, the factor oracle (rarely also called suffix oracle).
It may recognize (accept) more strings than substrings (suffixes) of the reversed pattern, but is much easier to construct.
It still guarantees that, once the FAIL state is reached, the sequence of read characters is not a substring of the reversed pattern.

The cost and shift functions are defined as for BDM, but based on the oracle.
We refer to~\cite{Navarro2002} for the construction details and further properties of the oracle.
By construction, BOM never gives longer shifts than B(N)DM.
The main advantage of BOM over BDM is reduced space usage and preprocessing time; the factor oracle only has $m+1$ states and can be constructed faster than a suffix automaton.


\section{Deterministic Arithmetic Automata}
\label{sec:daa}


In this section, we introduce deterministic arithmetic automata (DAAs).
They extend ordinary deterministic finite automata (DFAs) by performing a computation while one moves from state to state.
Even though DAAs can be shown to be formally equivalent to families of DFAs on an appropriately defined larger state space,
they are a useful concept before introducing probabilistic arithmetic automata (PAAs) and allow us to construct PAAs for the analysis of pattern matching algorithms in a simpler way.

\begin{definition}[Deterministic Arithmetic Automaton, DAA]
A \emph{deterministic arithmetic automaton} is a tuple
\[\daa = \big( \stateset, \state_0, \Sigma, \delta, \valset, \val_0, \emiset, (\daaemi_\state)_{\state\in \stateset},(\op_\state)_{\state\in \stateset} \big),\] 
where~$\stateset$ is a finite set of states, $\state_0\in \stateset$ is the start state, 
$\Sigma$ is a finite alphabet, $\delta: \stateset\times\Sigma\to \stateset$ is a transition function, 
$\valset$ is a finite or countable set of values, $\val_0\in\valset$ is called the start value, 
$\emiset$ is a finite set of emissions, $\daaemi_\state\in \emiset$ is the emission associated to state~$\state$, 
and $\op_\state: \valset\times \emiset\to\valset$ is a binary operation associated to state~$\state$.
\end{definition}

Informally, a DAA starts with the state-value pair~$(\state_0,\val_0)$ and reads a sequence of symbols from~$\Sigma$.
Being in state~$\state$ with value~$\val$, upon reading~$\sigma\in\Sigma$, the DAA performs a state transition to~$\state':=\delta(\state,\sigma)$ and updates the value to~$\val':=\op_{\state'}(\val,\daaemi_{\state'})$ using the operation and emission of the new state~$\state'$.

Further, we define the associated joint transition function
\[ \hat{\delta}: (\stateset\times\valset)\times \Sigma \to (\stateset\times\valset), \qquad
       \hat{\delta} \big( (\state,\val),\sigma \big) := \big(\delta(\state,\sigma)\,,\, \op_{\delta(\state,\sigma)}(\val,\daaemi_{\delta(\state,\sigma)})\big).
\]

As usual, we extend the definition of $\hat{\delta}$ inductively from $\Sigma$ to $\Sigma^*$ in its second argument by 
$\hat{\delta}\big((\state,\val),\emptystring\big):=(\state,\val)$ for the empty string~$\emptystring$ and
$\hat{\delta}\big((\state,\val),x\sigma\big) := \hat{\delta}\big(\hat{\delta}((\state,\val),x),\sigma\big)$ for all~$x\in\Sigma^*$ and~$\sigma\in\Sigma$.
When~$\hat{\delta}\big((\state_0,\val_0),s\big) = (\state,\val)$ for some $\state\in \stateset$ and $s\in\Sigma^*$, we say that~$\daa$ computes value~$\val$ for input~
$s$ and define~$\daaval_{\daa}(s) := \val$.

For each state~$\state$, the emission~$\daaemi_\state$ is fixed and could be dropped from the definition of DAAs. 
In fact, one could also dispense with values and operations entirely and define a DFA over state space
$\stateset\times\valset$, performing the same operations as a DAA. 
However, we intentionally include values, operations, and emissions to emphasize the connection to PAAs (which are defined in Section~\ref{sec:paa}).

As a simple example for a DAA, take a standard DFA $(\stateset,\state_0,\Sigma,\delta,F)$ with~$F\subset \stateset$ being a set of final (or accepting) states.
To obtain a DAA that counts how many times the DFA visits an accepting state when reading~$x\in\Sigma^\ast$, 
let~$\emiset:=\{0,1\}$ and define~$\daaemi_\state:=1$ if~$\state\in F$, and~$\daaemi_\state:=0$ otherwise. 
Further define~$\valset=\setN$ with~$\val_0:=0$, and let the operation in each state be the usual addition: $\op_\state(\val,\emi):=\val+\emi$ for all~$\state$.
Then~$\daaval_\daa(x)$ is the desired count.


\section{Constructing DAAs for Pattern Matching Analysis}
\label{sec:construction}
For a given algorithm and pattern $p\in\Sigma^m$ with known shift and cost functions, $\fshift^p: \Sigma^m\to\{1,\ldots,m\}$, $w\mapsto\shift{p}{w}$ and $\fcost^p: \Sigma^m\to\setN$, $w\mapsto\cost{p}{w}$, we construct a DAA that upon reading a text $s\in\Sigma^\ast$ computes the total cost, defined as the sum of costs over all examined windows.
(Which windows are examined depends of course on the shift values of previously examined windows.)
Slightly abusing notation, we write $\cost{p}{s}$ for the total cost incurred on $s$.

While different constructions are possible (see also~\cite{Marschall2010}), the construction presented here has the advantage that it is simple to describe and implement and processes only one text character at a time. This property allows the construction of a product DAA that directly compares two algorithms as detailed in Section~\ref{sec:comparison}.

We define a DAA by
\begin{itemize}
\item $\stateset := \Sigma^m \times \{0,\dots,m\}$,
\item $\state_0  := (p,m)$,
\item $\valset   := \setN$,
\item $\val_0    := 0$,
\item $\emiset   := \{1,\dots,m\}$,
\item $\daaemi_{(w,x)} := \begin{cases}
	0 &\text{ if } x>0,\\
	\cost{p}{w} &\text{ if } x=0,
	\end{cases}$
\item $\op_q     : (\val,\emi) \mapsto \val+\emi$ for all $q\in\stateset$ (addition),
\item $\delta    : \big((w,x),\sigma\big) \mapsto \begin{cases}
	(\suff{w}\sigma,\, x-1) &\text{ if } x>0,\\
	(\suff{w}\sigma,\, \shift{p}{w}-1) &\text{ if } x=0,
	\end{cases}$\\
where $\suff{w}$ is the length-$(m-1)$ suffix of $w$, i.e.,
$\suff{w}:=\chr{w}{1}\dots\chr{w}{m-1}$.
\end{itemize}
% ...
Informally, the state $\state=(w,x)$ means that the last $m$~read characters spell~$w$ and that~$x$ more characters need to be read to get to the end of the current window.
For the start state $(p,m)$, the component~$p$ is arbitrary, as we need to read $m$ characters to reach the end of the first window.
The value accumulates the cost of examined windows.
Therefore, the operation is a simple addition in each state, and the emission of state $(w,x)$ specifies the cost to add. 
Consequently, the emission is zero if the state does not correspond to an examined window ($x>0$), and the emission equals the window cost $\cost{p}{w}$ if $x=0$.
The transition function $\delta$ specifies how to move from one state to the next when reading the next text character $\sigma\in\Sigma$: 
In any case, the window content is updated by forgetting the first character and appending the read~$\sigma$.
If the end of the current window has not been reached ($x>0$), the counter $x$ is decremented. 
Otherwise, the window's shift value is used to compute the number of characters till the next window aligns.

\begin{theorem}
With the DAA $\daa$ constructed as above, $\daaval_\daa(s)=\cost{p}{s}$ for all $s\in\Sigma^*$.
\end{theorem}
\begin{proof}
The total cost $\cost{p}{s}$ can be written as the sum of costs of all processed windows: $\cost{p}{s}=\sum_{i\in I_s}\cost{p}{\substr{s}{i-m+1}{i}}$, where $I_s$ is the set of indices giving the processed windows, i.e.\ $I_s\subset\{m-1, \dots, \len{s}-1\}$ such that
\[i\in I_s \quad :\Longleftrightarrow \quad  i = m-1 \quad\text{or}\quad 
                 \exists j\in I_s: i = j+\shift{p}{\substr{s}{j-m+1}{j}}.\]
We have to prove that the DAA computes this value for $s\in\Sigma^*$. 

Let $(w_i,x_i)$ be the DAA state active after reading $\prefix{s}{i}$.
Observe that the transition function $\delta$ ensures that the $w_i$-component of $(w_i,x_i)$ reflects the rightmost length-$m$ window of $\prefix{s}{i}$, which can immediately be verified inductively. 
Thus, the emission on reading the last character $\chr{s}{i}$ of $\prefix{s}{i}$ with $i\geq m-1$ is, by definition of $\daaemi_{(w_i,x_i)}$, either $\cost{p}{\substr{s}{i-m+1}{i}}$ or zero, depending on the second component of $(w_i,x_i)$. 
As the operation is an addition for all states, $\daaval_\daa(s)=\sum_{i\in I'_s}\, \cost{p}{\substr{s}{i-m+1}{i}}$ for 
\[ I'_s := \big\{i\in\{0,\ldots,\len{s}-1\}: x_i=0\big\}. \]

It remains to show that $I_s=I'_s$. 
To this end, note that by $\delta$, we have $x_{i+1}=x_i-1$ if $x_{i+1}>0$ and $x_{i+1}=\shift{p}{w_i}-1$ if $x_{i+1}=0$. 
As $\state_0=(p,m)$, it follows that $m-1\in I'_s$. 
Using $w_i=\substr{s}{i-m+1}{i}$ for $i\geq m-1$, we conclude that whenever $x_i=0$, it follows that $x_j=0$ for $j=i+\shift{p}{\substr{s}{i-m+1}{i}}$ and that $x_{j'}>0$ for $i<j'<j$.
Hence we obtain that $i\in I'_s$ implies that
$i+\shift{p}{\substr{s}{i-m+1}{i}}\in I'_s$ and $i+k\notin I'_s$ for $0<k<\shift{p}{\substr{s}{i-m+1}{i}}$,
which completes the proof.
\end{proof}


\paragraph{DAA Minimization}
The size of the constructed DAA's state space depends exponentially on the pattern length, making the application for long patterns infeasible in practice.
However, depending on the particular circumstances (i.e., algorithm and pattern analyzed), the constructed DAA can often be substantially reduced by applying a modified version of Hopcroft's algorithm for DFA minimization~\cite{Hopcroft1971}.

Hopcroft's algorithm minimizes a DFA in $O(|\stateset|\log |\stateset|)$ time by iteratively refining a partition of the state set. 
In the beginning, all states are partitioned into two distinct sets: one containing the accepting stats, and the other containing the non-accepting states.
This partition is iteratively refined whenever a reason for non-equivalence of two states in the same set is found.
Upon termination, the states are partitioned into sets of equivalent states. 
Refer to~\cite{Knuutila2001} for an in-depth explanation of Hopcroft's algorithm.

The algorithm can straightforwardly be adapted to minimize DAAs by choosing the initial state set partition appropriately. 
In our case, each DAA state is associated with the same operation. 
The only differences in state's behavior thus stem from different emissions. 
Therefore, Hopcroft's algorithm can be initialized by the partition induced by the emissions and then continued as usual.

As we exemplify in Section~\ref{sec:case_studies}, this leads to a considerable reduction of the number of states.


\section{Probabilistic Arithmetic Automata}
\label{sec:paa}

This section introduces finite-memory random text models and explains how to construct a \emph{probabilistic arithmetic automaton} (PAA) from a (minimized) DAA and a random text model.
PAAs were introduced in~\cite{Marschall2008}, where they  are used to compute pattern occurrence count distributions.
Other applications in biological sequence analysis include the exact computation of p-values of sequence motifs~\cite{MarschallRahmann2009EfficientExactMotifDiscovery}, 
and the determination of seed sensitivity for pairwise sequence alignment algorithms based on filtering~\cite{HermsRahmann2008ComputingSeedSensitivity}.

\subsection{Random Text Models}\label{sec:text_model}
Given an alphabet~$\Sigma$, a random text is a stochastic process~$(S_t)_{t\in\setN_0}$, where each $S_t$ takes values in~$\Sigma$. A text model $\prob$ is a probability measure assigning probabilities to (sets of) strings. 
It is given by (consistently) specifying the probabilities $\prob(S_0\ldots S_{\len{s}-1}=s)$ for all~$s\in\Sigma^\ast$. 
We only consider finite-memory models in this article which are formalized in the following definition.

\begin{definition}[Finite-memory text model]\label{def:text_model}
A finite-memory text model is a tuple $(\mathcal{C},c_0,\Sigma,\varphi)$, where $\mathcal{C}$ is a finite state space (called \emph{context space}), $c_0\in\mathcal{C}$ a start context, $\Sigma$ an alphabet, and $\varphi: \mathcal{C}\times\Sigma\times\mathcal{C}\to[0,1]$ a transition function with $\sum_{\sigma\in\Sigma,c'\in\mathcal{C}}\varphi(c,\sigma,c')=1$ for all $c\in\mathcal{C}$. The random variable giving the text model state after $t$~steps is denoted~$C_t$ with $C_0:\equiv c_0$.
A probability measure is now induced by stipulating
\[ \prob(S_0\ldots S_{\totalsteps-1}=s,C_1=c_1,\ldots,C_\totalsteps=c_\totalsteps) 
  := \prod_{i=0}^{\totalsteps-1}\, \varphi(c_{i},\chr{s}{i},c_{i+1})
\]
for all $\totalsteps\in\setN_0$, $s\in\Sigma^\totalsteps$, and $(c_1,\ldots,c_\totalsteps)\in\mathcal{C}^\totalsteps$.
\end{definition}

The idea is that the model given by $(\mathcal{C},c_0,\Sigma,\varphi)$ generates a random text by moving from context to context and emitting a character at each transition, where $\varphi(c,\sigma,c')$ is the probability of moving from context~$c$ to context~$c'$ and thereby generating the letter~$\sigma$.

Note that the probability $\prob(S_0\ldots S_{\len{s}-1}=s)$ is obtained by marginalization over all context sequences that generate $s$.
This can be efficiently done, using the decomposition of the following lemma.

\begin{lemma}\label{lem:text_model}
Let $(\mathcal{C},c_0,\Sigma,\varphi)$ be a finite-memory text model. Then, 
\[ \prob(S_0\ldots S_{\totalsteps}=s\sigma,C_{\totalsteps+1}=c)
   = \sum_{c'\in\mathcal{C}}\, \prob(S_0\ldots S_{\totalsteps-1}=s,C_\totalsteps=c') \cdot \varphi(c',\sigma,c)
\]
for all $n\in\setN_0$, $s\in\Sigma^n$, $\sigma\in\Sigma$ and $c\in\mathcal{C}$.
\end{lemma}
\begin{proof}
We have
\begin{align*}
 & \prob(S_0\ldots S_{\totalsteps}=s\sigma,C_{\totalsteps+1}=c) \\
=& \sum_{c_1,\ldots,c_\totalsteps}\, \prob(S_0\ldots S_{\totalsteps}=s\sigma, C_1=c_1,\ldots,C_{\totalsteps}=c_\totalsteps,C_{\totalsteps+1}=c) \\
=& \sum_{c_1,\ldots,c_\totalsteps}\, \prod_{i=0}^{\totalsteps-1}\, \varphi(c_{i},\chr{s}{i},c_{i+1})\cdot\varphi(c_\totalsteps,\sigma,c) \\
=& \sum_{c_{\totalsteps}\in\mathcal{C}}\, 
   \left(\sum_{c_1,\ldots,c_{\totalsteps-1}}\, \prod_{i=0}^{\totalsteps-1}\,   
   \varphi(c_{i},\chr{s}{i},c_{i+1})\right) \cdot \varphi(c_\totalsteps,\sigma,c) \\
=& \sum_{c_{\totalsteps}\in\mathcal{C}}\prob(S_0\ldots S_{\totalsteps-1}=s,C_{\totalsteps}=c_n)  \cdot\varphi(c_\totalsteps,\sigma,c)\,.
\end{align*}
Renaming $c_n$ to $c'$ yields the claimed result.
\end{proof}

Similar text models are used in~\cite{Kucherov2006}, where they a called probability transducers. 
In the following, we refer to a finite-memory text model $(\mathcal{C},c_0,\Sigma,\varphi)$ simply as text model, as all text models considered in this article are special cases of Definition~\ref{def:text_model}.

For an i.i.d.\ model, we set $\mathcal{C}=\{\emptystring\}$ and $\varphi(\emptystring,\sigma,\emptystring)=p_\sigma$ for each~$\sigma\in\Sigma$, where~$p_\sigma$ is the occurrence probability of letter~$\sigma$ (and $\emptystring$ may be interpreted as an empty context).
For a Markovian text model of order~$r$, the distribution of the next character depends on the~$r$ preceding characters (fewer at the beginning); thus $\mathcal{C}:=\bigcup_{i=0}^{r}\Sigma^i$.
This notion of text models also covers variable order Markov chains as introduced in~\cite{Schulz2008}, which can be converted into equivalent models of fixed order. Text models as defined above have the same expressive power as character-emitting HMMs, that means, they allow to construct the same probability distributions.

%

\subsection{Basic PAA Concepts}\label{sec:paa_def}
Probabilistic arithmetic automata (PAAs), as introducted in~\cite{Marschall2008}, are a generic concept useful to model probabilistic chains of operations. In this section, we sum up the definition and basic recurrences needed in this article.

\begin{definition}[Probabilistic Arithmetic Automaton, PAA]\label{def:paa}
A \emph{probabilistic arithmetic automaton} is a tuple $\paa = \paatuple$, where $\stateset$, $\state_0$, $\valset$, $\val_0$, $\emiset$ and $\op$ have the same meaning as for a DAA, each~$\emidist_\state$ is a state-specific probability distribution on the emissions~$\emiset$, and~$\paatrans:\stateset\times\stateset\to[0,1]$ is a transition function, such that~$\paatrans(\state,\state')$ gives the probability of a transition from state~$\state$ to state~$\state'$, i.e.\ $\big(\paatrans(\state,\state')\big)_{\state,\state'\in\stateset}$ is a stochastic matrix.

A PAA induces three stochastic processes:
(1) the state process $(\stateproc_t)_{t\in\setN}$ with values in~$\stateset$,
(2) the emission process $(\emiproc_t)_{t\in\setN}$ with values in~$\emiset$, and
(3) the value process $(\valproc_t)_{t\in\setN}$ with values in~$\valset$ such that
$\valset_0 :\equiv \val_0$ and $\valset_t := \op_{\stateproc_t}\left(\valproc_{t-1}, \emiproc_t\right).$
\end{definition}

We now restate the PAA recurrences from~\cite{Marschall2008} to compute the state-value distribution after~$t$ steps.
For the sake of a shorter notation, we define~$f_{t}(\state,\val):=\prob(\stateproc_t=\state, \valproc_t=v)$.
Since we are generally only interested in the value distribution, note that it can be obtained by marginalization:
$\prob(\valproc_t=\val) = \sum_{\state\in\stateset}\, f_{t}(\state,\val)$.

\begin{lemma}[State-value recurrence, \cite{Marschall2008}]\label{lem:paa}\hfill\\
The state-value distribution is given by
$f_0(\state,\val) = 1$ if~$\state=\state_0$ and~$\val=\val_0$, and $f_0(\state,\val)=0$ otherwise. 
For~$t\geq 0$,
\begin{equation}\label{eqn:paa_recurrence}
  f_{t+1}(\state,\val) = 
  \sum_{\state'\in\stateset}\, \sum_{(\val',\emi)\in\op^{-1}_\state(\val)}\,
   f_t(\state',\val')\cdot \paatrans(\state',\state)\cdot \emidist_\state(\emi),
\end{equation}
where~$\op^{-1}_{\state}(\val)$ denotes the inverse image set of~$\val$ under~$\op_\state$.
\end{lemma}
The recurrence in Lemma~\ref{lem:paa} resembles the Forward recurrences known from HMMs.

Note that the range of $\valproc_t$ is finite for each $t$, even when $\valset$ is infinite, as $\valproc_t$ is a function of the states and emissions up to time $t$, and state set $\stateset$ and emission set $\emiset$ are finite.
We define $\valset_t := \text{range } \valproc_t$ and $\valsize_\totalsteps := \max_{0\leq t\leq \totalsteps}\, |\valset_t|$.
Clearly $\valsize_\totalsteps \leq (|\stateset|\cdot|\emiset|)^\totalsteps$. 
Therefore all actual computations are on finite sets.
When analyzing the number of character accesses of a pattern matching algorithm, we have $\valset_t\subset\{0,\ldots,m(\totalsteps-m+1)\}$, as  at most $(\totalsteps-m+1)$ search windows are processed, each causing at most $m$ character accesses. Thus, $\valsize_\totalsteps\in O(\totalsteps\cdot m)$.

\subsection{Constructing a PAA from a DAA and a Text Model}\label{sec:daapaa}
We now formally state how to combine a DAA and a text model into a PAA that allows us to compute the distribution of values produced by the DAA when processing a random text.

\begin{lemma}[DAA $+$ Text model $\to$ PAA]\label{lem:daapaa}
Let $(\mathcal{C},c_0,\Sigma,\varphi)$ be a text model and $\daa = \big(\stateset^\daa, \state_0^\daa, \Sigma, \delta, \valset, \val_0, \emiset, (\daaemi_\state)_{\state\in \stateset^\daa}, (\op_\state^\daa)_{\state\in \stateset^\daa} \big)$ be a DAA. Then, define
\begin{itemize}
 \item a state space $\stateset:=\stateset^\daa\times\mathcal{C}$,
 \item a start state~$\state_0:=(\state_0^\daa,c_0)$,
 \item transition probabilities
 \begin{equation}\label{eqn:daa_paa_transfunc}
 \paatrans\big((\state,c),(\state',c')\big):=\sum_{\sigma\in\Sigma:\,\delta(\state,\sigma)=\state'}\varphi(c,\sigma,c'),
 \end{equation}
 \item (deterministic) emission probability vectors
       \[\emidist_{(\state,c)}(\emi):=
       \begin{cases}
       1 & \mbox{if } \emi=\daaemi_\state\,, \\
       0 & \mbox{otherwise}\,,
       \end{cases}\]
       for all $(\state,c)\in \stateset$.
 \item operations $\op_{(\state,c)}(\val,\emi):=\op^\daa_\state(\val,\emi)$ for all $(\state,c)\in \stateset$.
\end{itemize}
Then, $\paa = \paatuple$ is a PAA with
\[\dist(\valproc_t) = \dist\big(\daaval_{\daa}(S_{0}\dots S_{t-1})\big),\]
for all $t\in\setN_0$, where~$S$ is a random text according to the text model $(\mathcal{C},c_0,\Sigma,\varphi)$. 
States having zero probability of being reached from $\state_0$ may be omitted from $\stateset$ and $\paatrans$. 
For such a PAA, the value distribution $\dist(\valproc_\totalsteps)$ can be computed with $O(\totalsteps\cdot |\stateset^\daa|\cdot|\mathcal{C}|^2\cdot |\Sigma| \cdot \valsize_\totalsteps)$ operations using $O(|\stateset^\daa|\cdot|\mathcal{C}|\cdot \valsize_\totalsteps)$ space. 
If for all $c\in\mathcal{C}$ and $\sigma\in\Sigma$, there exists at most one $c'\in\mathcal{C}$ such that $\varphi(c,\sigma,c')>0$, then the runtime is bounded by $O(\totalsteps\cdot |\stateset^\daa|\cdot|\mathcal{C}|\cdot |\Sigma| \cdot \valsize_\totalsteps)$.
\end{lemma}

\begin{proof} 
$\paa$ is a PAA by Definition~\ref{def:paa}.
As in Section~\ref{sec:paa_def}, we define $f_{t}(\state,\val) := \prob(\stateproc_t=\state,\valproc_t=\val)$.
Iverson brackets are written $\iverl\cdot\iverr$, i.e.\ $\iverl A\iverr=1$ if the statement~$A$ is true and $\iverl A\iverr=0$ otherwise. 

To prove $\dist(\valproc_{t}) = \dist\big(\daaval_{\daa}(S_{0}\dots S_{{t}-1})\big)$, we show that 
\begin{equation}\label{eqn:paadaa}
f_{t}\big((\state^\daa,c),\val\big) = \sum_{s\in\Sigma^{t}}\big\iverl\hat{\delta}\big((\state_0^\daa,\val_0),s\big)=(\state^\daa,\val)\big\iverr\cdot\prob(S_0\ldots S_{t-1}=s,C_t=c)
\end{equation}
for all $\state^\daa\in \stateset^\daa$, $c\in\mathcal{C}$, $\val\in\valset$, and $t\in\setN_0$. For~$t=0$, Equation~\eqref{eqn:paadaa} is correct by definitions of PAAs, DAAs and text models.
For $t>0$ we prove it inductively.
Assume~\eqref{eqn:paadaa} to be correct for all $t'$ with $0\leq t'<t$.
Then
\begin{align}
  & f_{t}\big(\underbrace{(\state^\daa,c)}_{=:\state},\val\big)  \label{eqn:paadaa1} \\
= & \sum_{\state'\in \stateset}\; \sum_{(\val',\emi)\in\op^{-1}_\state(\val)}\, 
     f_{t-1}(\state',\val')\cdot \paatrans(\state',\state)\cdot \emidist_\state(\emi)  \label{eqn:paadaa2}\\
= & \sum_{\state'\in \stateset}\; \sum_{(\val',\emi)\in\valset\times \emiset}\, 
    \big\iverl\op^\daa_{\state^\daa}(\val',\emi)=\val\big\iverr \cdot f_{t-1}(\state',\val')
    \cdot \paatrans(\state',\state)\cdot\big\iverl\daaemi_{\state^\daa}=\emi\big\iverr \label{eqn:paadaa3}\\
\begin{split}
= & \sum_{\state'^\daa\in \stateset^\daa}\; \sum_{c'\in\mathcal{C}}\; \sum_{(\val',\emi)\in\valset\times \emiset}\,
    \big\iverl\op^\daa_{\state^\daa}(\val',\emi)=\val\big\iverr  \cdot \big\iverl\daaemi_{\state^\daa}=\emi\big\iverr 
    \cdot f_{t-1}(\state',\val')\\
    & \qquad \cdot \sum_{\sigma\in\Sigma}\, \big\iverl\delta(\state'^\daa,\sigma)=\state^\daa\big\iverr
    \cdot \varphi(c',\sigma,c)
\end{split}\label{eqn:paadaa4}\\
\begin{split}
= & \sum_{s\in\Sigma^{t-1}}\; \sum_{\sigma\in\Sigma}\; \sum_{\state'^\daa\in \stateset^\daa}\;
    \sum_{c'\in \mathcal{C}}\; \sum_{(\val',\emi)\in\valset\times \emiset}\,
    \big\iverl \op^\daa_{\state^\daa}(\val',\emi)=\val\big\iverr \cdot 
    \big\iverl \daaemi_{\state^\daa}=\emi\big\iverr \\
& \qquad \cdot \big\iverl \delta(\state'^\daa,\sigma)=\state^\daa\big\iverr \cdot
    \big\iverl \hat{\delta}\big((\state_0^\daa,\val_0),s\big)=(\state'^\daa,\val')\big\iverr \\
& \qquad \cdot \prob(S_0\ldots S_{t-2}=s,C^{t-1}=c') \cdot \varphi(c',\sigma,c)
\end{split}  \label{eqn:paadaa5}\\
\begin{split}
= & \sum_{s\sigma\in\Sigma^t}\; \sum_{\state'^\daa\in \stateset^\daa}\;
    \sum_{(\val',\emi)\in\valset\times \emiset}\, 
    \big\iverl \op^\daa_{\state^\daa}(\val',\emi)=\val\big\iverr  \cdot
    \big\iverl \daaemi_{\state^\daa}=\emi\big\iverr \cdot
    \big\iverl\hat{\delta}\big((\state_0^\daa,\val_0),s\big)=(\state'^\daa,\val')\big\iverr \\
& \qquad \cdot \big\iverl \delta(\state'^\daa,\sigma)=\state^\daa\big\iverr
    \cdot \prob(S_0\ldots S_{t-1}=s\sigma,C_t=c)
\end{split}\label{eqn:paadaa6}\\
= & \sum_{s\sigma\in\Sigma^t}\, 
    \big\iverl\hat{\delta}\big((\state_0^\daa,\val_0),s\sigma\big)=(\state^\daa,\val)\big\iverr \cdot
    \prob(S_0\ldots S_{t-1}=s\sigma,C_t=c)  \label{eqn:paadaa7}
\end{align}
In the above derivation, step \eqref{eqn:paadaa1}$\to$\eqref{eqn:paadaa2} follows from~\eqref{eqn:paa_recurrence}.
Step \eqref{eqn:paadaa2}$\to$\eqref{eqn:paadaa3} follows from the definitions of $\op_\state$ and $\emidist_\state$.
Step \eqref{eqn:paadaa3}$\to$\eqref{eqn:paadaa4} uses the definitions of~$\paatrans$ and $\stateset$ in Lemma~\ref{lem:daapaa}.
Step \eqref{eqn:paadaa4}$\to$\eqref{eqn:paadaa5} uses the induction assumption.
Step \eqref{eqn:paadaa5}$\to$\eqref{eqn:paadaa6} uses Lemma~\ref{lem:text_model}.
The final step \eqref{eqn:paadaa6}$\to$\eqref{eqn:paadaa7} follows by combining the four Iverson brackets summed over $\state'^\daa$ and $(\val',\emi)$ into a single Iverson bracket.

To compute the table $f_\totalsteps$ containing $f_\totalsteps(\state,\val)$ for all $\state\in\stateset$ and $\val\in\valset$, we start with $f_0$ and perform $\totalsteps$ update steps. The given runtime bounds can be verified by considering a ``push'' algorithm: When computing $f_{t+1}$, we initialize the table with zeros and iterate over all $\state\in\stateset$, $\val\in\valset$ and $\state'\in\{\state''\in\stateset:\paatrans(\state,\state'')>0\}$; for each combination of $\state$, $\val$, and $\state'$, we add $f_t(\state,\val)\cdot\paatrans(\state,\state')$ to $f_{t+1}\big(\state',\op_{\state'}(\val,\daaemi_{\state'})\big)$.
\end{proof}

As a direct consequence of the above lemma and of the DAA construction from Section~\ref{sec:construction}, we arrive at our main theorem.

\begin{theorem}\label{thm:main}
Given a finite-memory text model $(\mathcal{C},c_0,\Sigma,\varphi)$,
a window-based pattern matching algorithm $A$, a pattern $p$ with $|p|=m$,
and the functions $\fshift^{A,p}$ and $\fcost^{A,p}$,
the cost distribution $\dist(X_\totalsteps^{A,p})$ can be computed using $O(\totalsteps^2\cdot m\cdot |\stateset^\daa|\cdot|\mathcal{C}|^2\cdot |\Sigma|)$ time and $O(|\stateset^\daa|\cdot|\mathcal{C}|\cdot\totalsteps\cdot m)$ space. 
Since $|\stateset^\daa|$ is bounded by $O(m\cdot\Sigma^m)$, the computation uses $O(\totalsteps^2\cdot m^2\cdot\Sigma^{m+1}\cdot|\mathcal{C}|^2)$ time and $O(m^2\cdot\Sigma^m\cdot|\mathcal{C}|\cdot\totalsteps)$ space.
If for all $c\in\mathcal{C}$ and $\sigma\in\Sigma$, there exists at most one $c'\in\mathcal{C}$ such that $\varphi(c,\sigma,c')>0$, a factor of $|\mathcal{C}|$ can be dropped from the runtime bounds.
\end{theorem}

Applying DAA minimization before the PAA construction results in considerable speed-ups in practice. Alternatively, algorithm dependent construction schemes may be used to construct small automata. Tsai~\cite{Tsai2006}, for instance, gives algorithms to compute the asymptotic mean and variance of the number of comparisons used by Horspool's algorithm; for that, he constructs a Markov chain with $O(m^2)$ states. His construction can immediately be adapted to construct a DAA with $O(m^2)$ states.


\section{Comparing Algorithms with Difference DAAs}
\label{sec:comparison}

Computing the cost distribution for two algorithms allows us to compare their performance characteristics.
One natural question, however, cannot be answered by comparing these two (one-dimensional) distributions: 
What is the probability that algorithm $A$ needs more text accesses than algorithm $B$ to scan the same random text? 
The answer will depend on the correlation of algorithm performances: Do the same instances lead to long runtimes for both algorithms or are there instances that are easy for one algorithm but difficult for the other? 
This section answers these questions by constructing a PAA to compute the distribution of \emph{cost differences} of two algorithms. 
That means, we calculate the probability that algorithm $A$ needs $v$ text accesses \emph{more} than algorithm $B$ for all $v\in\setZ$.

We start by giving a general construction of a DAA that computes the difference of the sum of emission of two given DAAs.

\begin{definition}[Difference DAA]\label{def:diff-daa}
Let $\daa^1=\left(\stateset^1, \state_0^1, \Sigma, \delta^1, \valset^1, \val_0^1, \emiset^1, (\daaemi^1_\state)_{\state\in \stateset^1},(\op^1_\state)_{\state\in \stateset^1} \right)$ and $\daa^2=\left(\stateset^2, \state_0^2, \Sigma, \delta^2, \valset^2, \val_0^2, \emiset^2, (\daaemi^2_\state)_{\state\in \stateset^2},(\op^2_\state)_{\state\in \stateset^2}\right)$ be DAAs given over the same alphabet~$\Sigma$ with $\valset^1=\valset^2=\setN$, $\val_0^1=\val_0^2=0$, $\emiset^1,\emiset^2\subset\setN$, and all operations are additions of previous value and current emission. The \emph{difference DAA} is defined as
\[\diffdaa{\daa^1}{\daa^2}:=\left(\stateset, \state_0, \Sigma, \delta, \valset, \val_0, \emiset, (\daaemi_\state)_{\state\in \stateset},(\op_\state)_{\state\in \stateset} \right)\,\]
where
\begin{itemize}
\item $\stateset:=\stateset^1\times\stateset^2$ and $\state_0:=(\state_0^1, \state_0^2)$,
\item $\valset:=\setZ$ and $\val_0:=0$,
\item $\emiset:=\emiset^1\times\emiset^2$ and $\daaemi_{(\state^1,\state^2)}:=\left(\daaemi^1_{\state^1},\daaemi^2_{\state^2}\right)$,
\item $\delta:\big((\state^1,\state^2),\sigma\big)\mapsto\big(\delta^1(\state^1,\sigma), \delta^2(\state^2,\sigma)\big)$,
\item $\op_\state:\big(\val,(\emi^1,\emi^2)\big)\mapsto\val+\emi^1-\emi^2$
\end{itemize}

\end{definition}

\begin{lemma}\label{lem:diffdaa}
Let $\daa^1$ and $\daa^2$ be DAAs meeting the criteria given in Definition~\ref{def:diff-daa} and $\daa:=\diffdaa{\daa^1}{\daa^2}$. Then,
\[\daaval_\daa(s)=\daaval_{\daa^1}(s)-\daaval_{\daa^2}(s)\]
for all $s\in\Sigma^*$.
\end{lemma}
\begin{proof}
Follows directly from Definition~\ref{def:diff-daa}.
\end{proof}

Lemma~\ref{lem:diffdaa} can now be applied to the DAAs constructed for the analysis of two algorithms as described in Section~\ref{sec:construction}. Since the above construction builds the product of both state spaces, it is advisable to 
minimize both DAAs before generating the product. Furthermore, in an implementation, only reachable states of the product automaton need to be  constructed. Before being used to build a PAA (by applying Lemma~\ref{lem:daapaa}), the product DAA should again be minimized.

As discussed in Section~\ref{sec:paa_def}, at most $m(n-m+1)$ character accesses can result from scanning a text of length $n$ for a pattern of length $m$. Thus, the difference of costs for two algorithms lies in the range $\{-m(n-m+1),\ldots,m(n-m+1)\}$ and, hence, $\vartheta_n\in O(n\cdot m)$.


\section{Case Studies}
\label{sec:case_studies}

\begin{figure}[t!]
\begin{center}
\includegraphics[angle=-90,width=.5\textwidth]{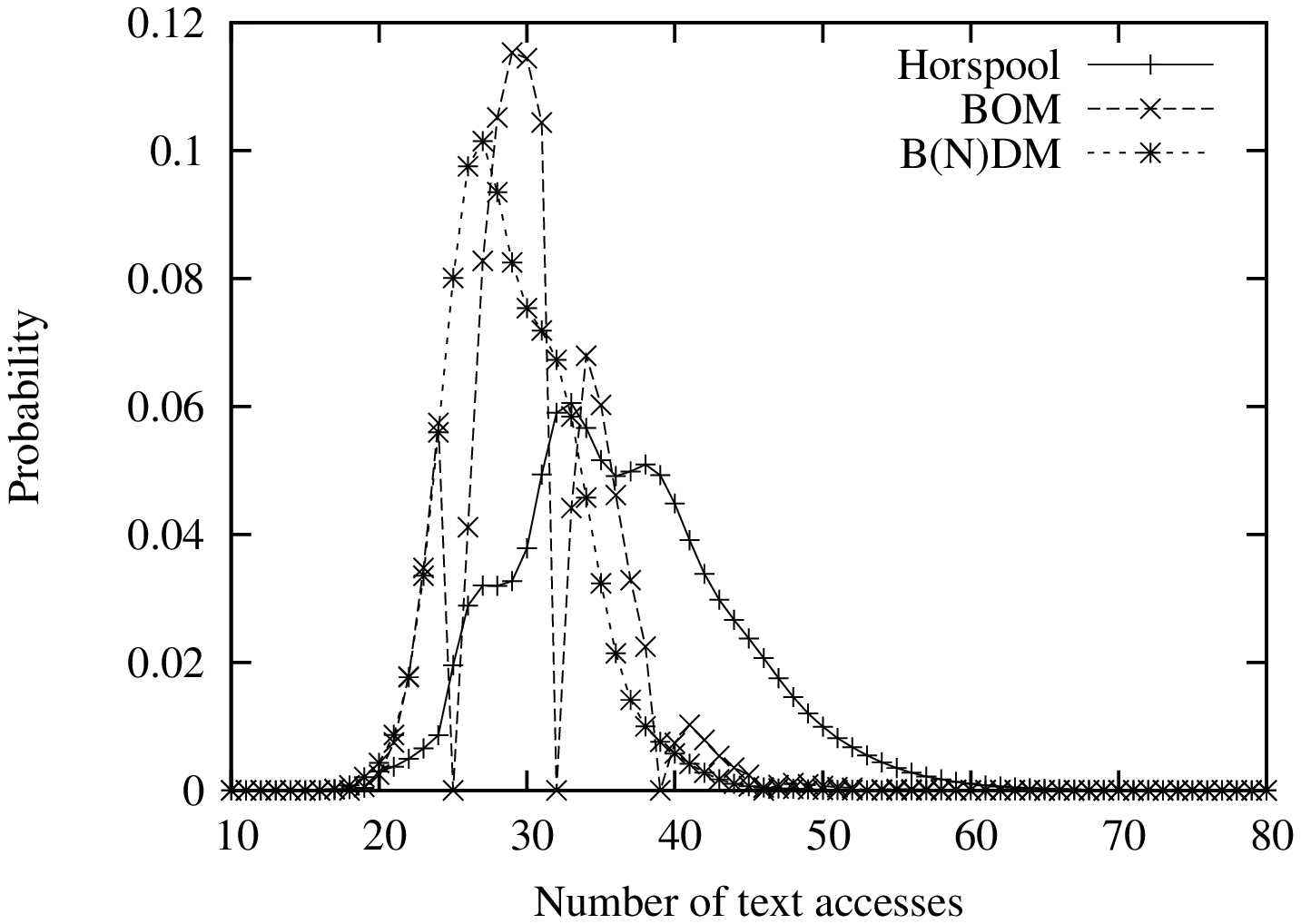}%
\includegraphics[angle=-90,width=.5\textwidth]{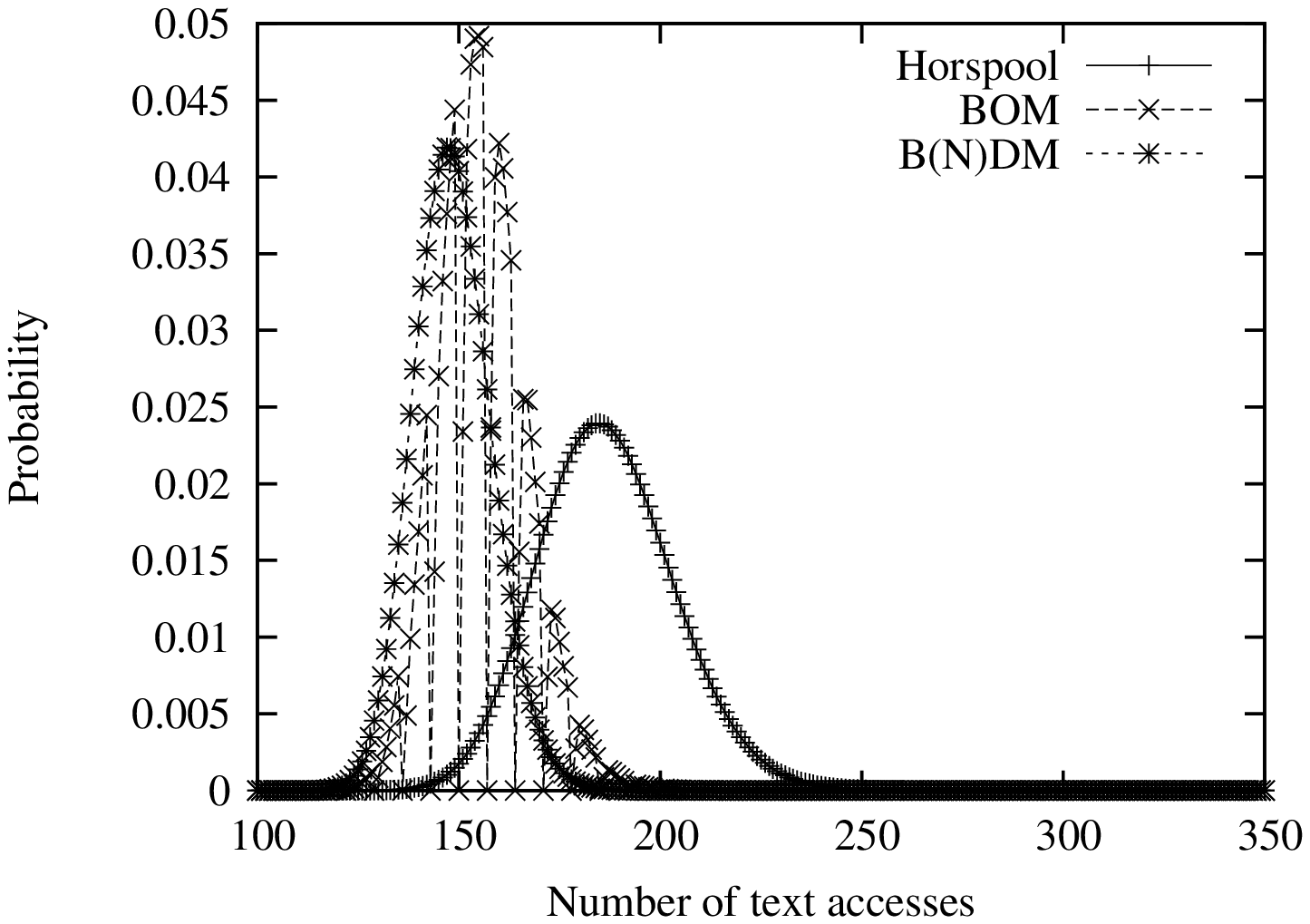}\\
\includegraphics[angle=-90,width=.5\textwidth]{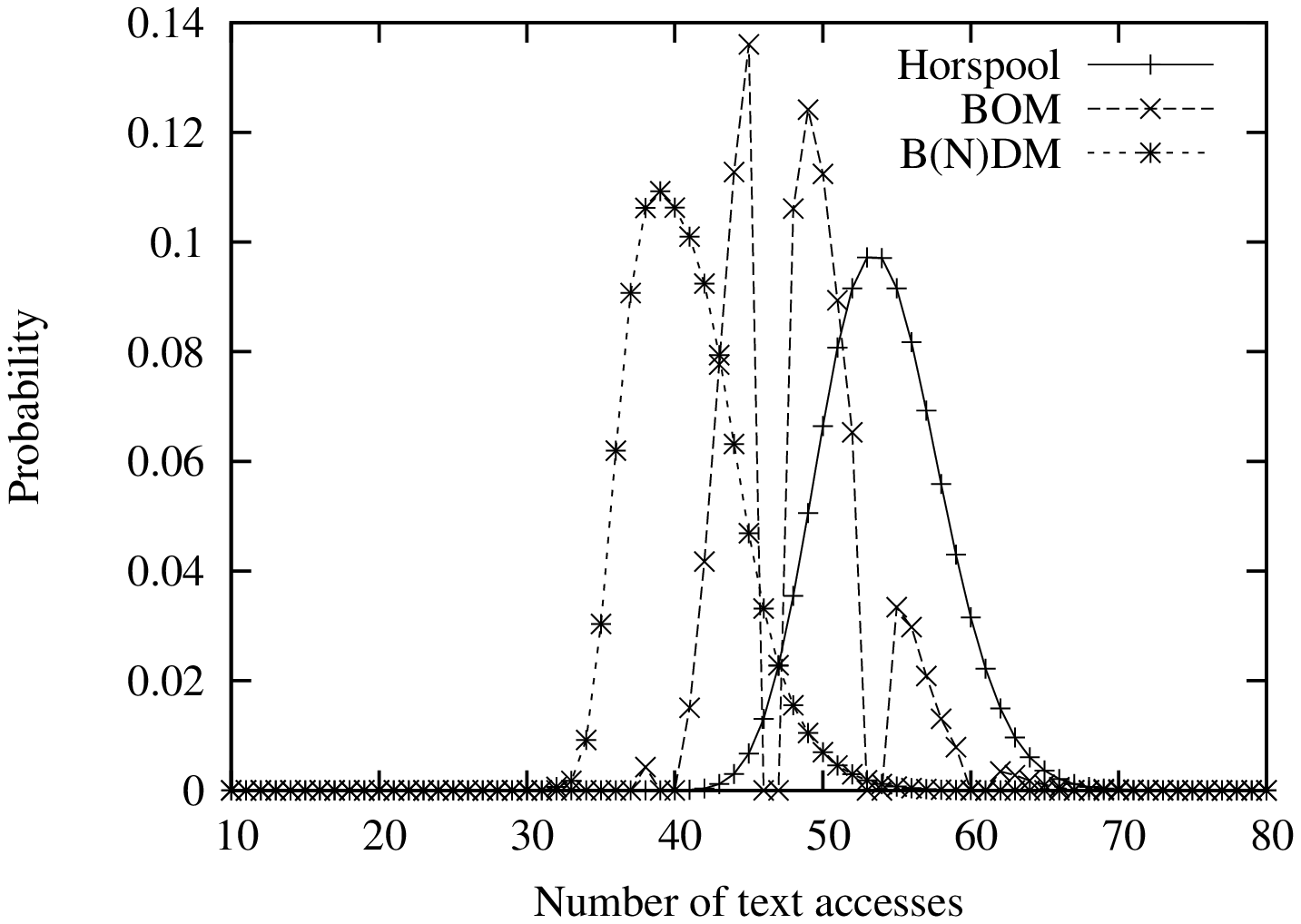}%
\includegraphics[angle=-90,width=.5\textwidth]{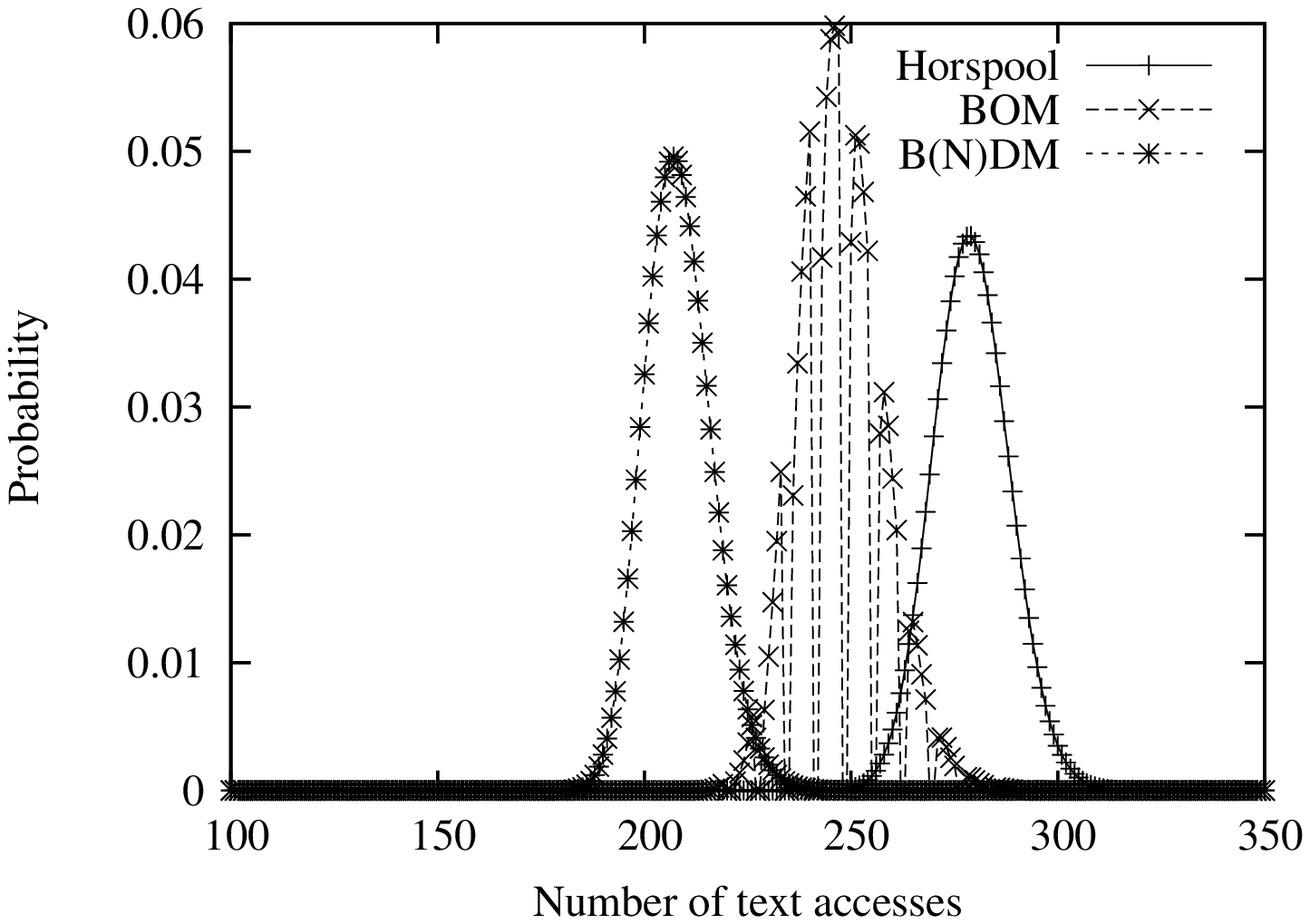}\\
\end{center}
\caption{Exact distributions of character access counts for patterns \texttt{ATATAT} (top) and \texttt{ACGTAC} (bottom) for text length 100 (left) and text length 500 (right). An i.i.d.\ text model with uniform character distribution is used.}\label{fig:dists}
\end{figure}

\begin{figure}[t!]
\begin{center}
\includegraphics[angle=-90,width=.5\textwidth]{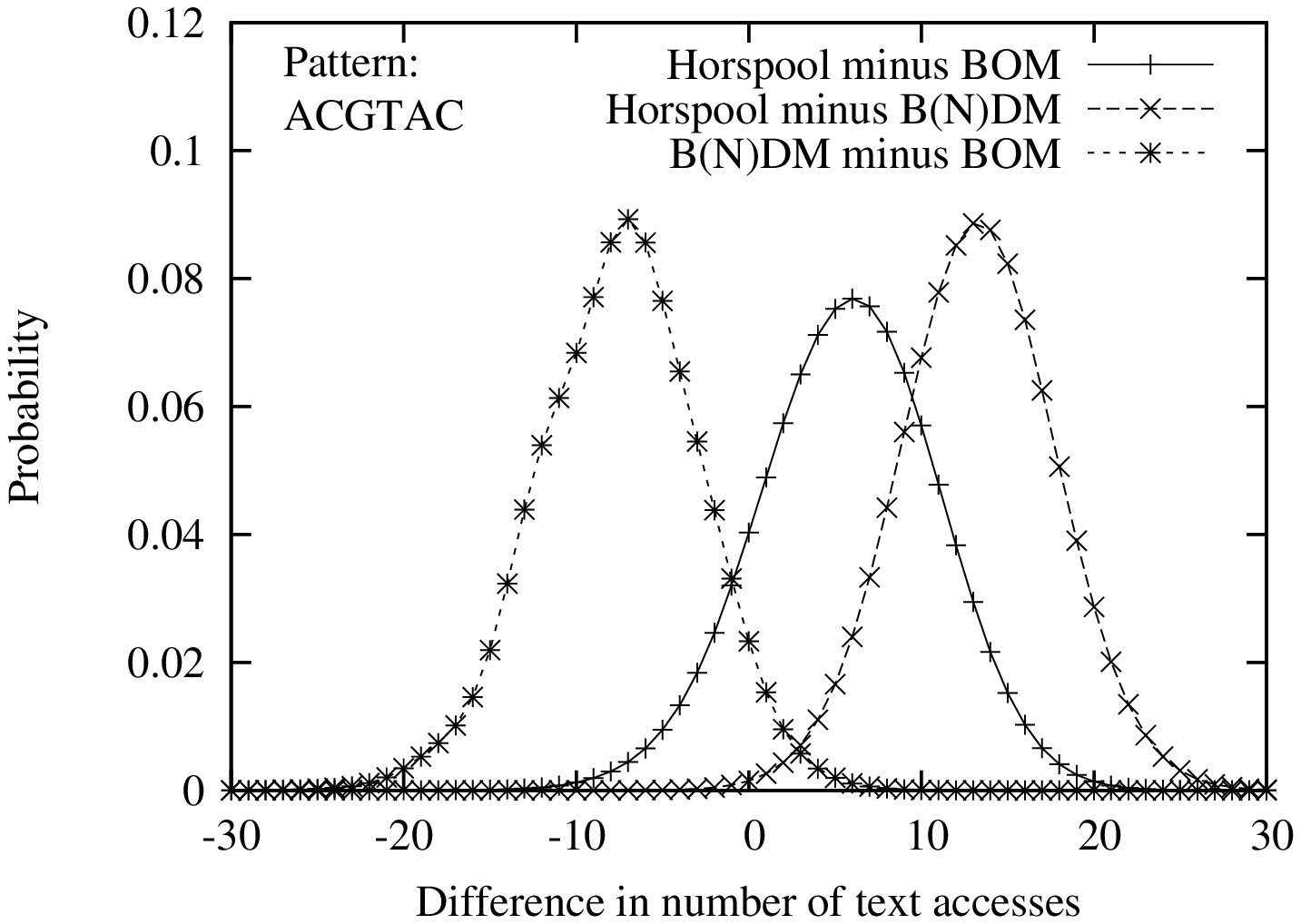}%
\includegraphics[angle=-90,width=.5\textwidth]{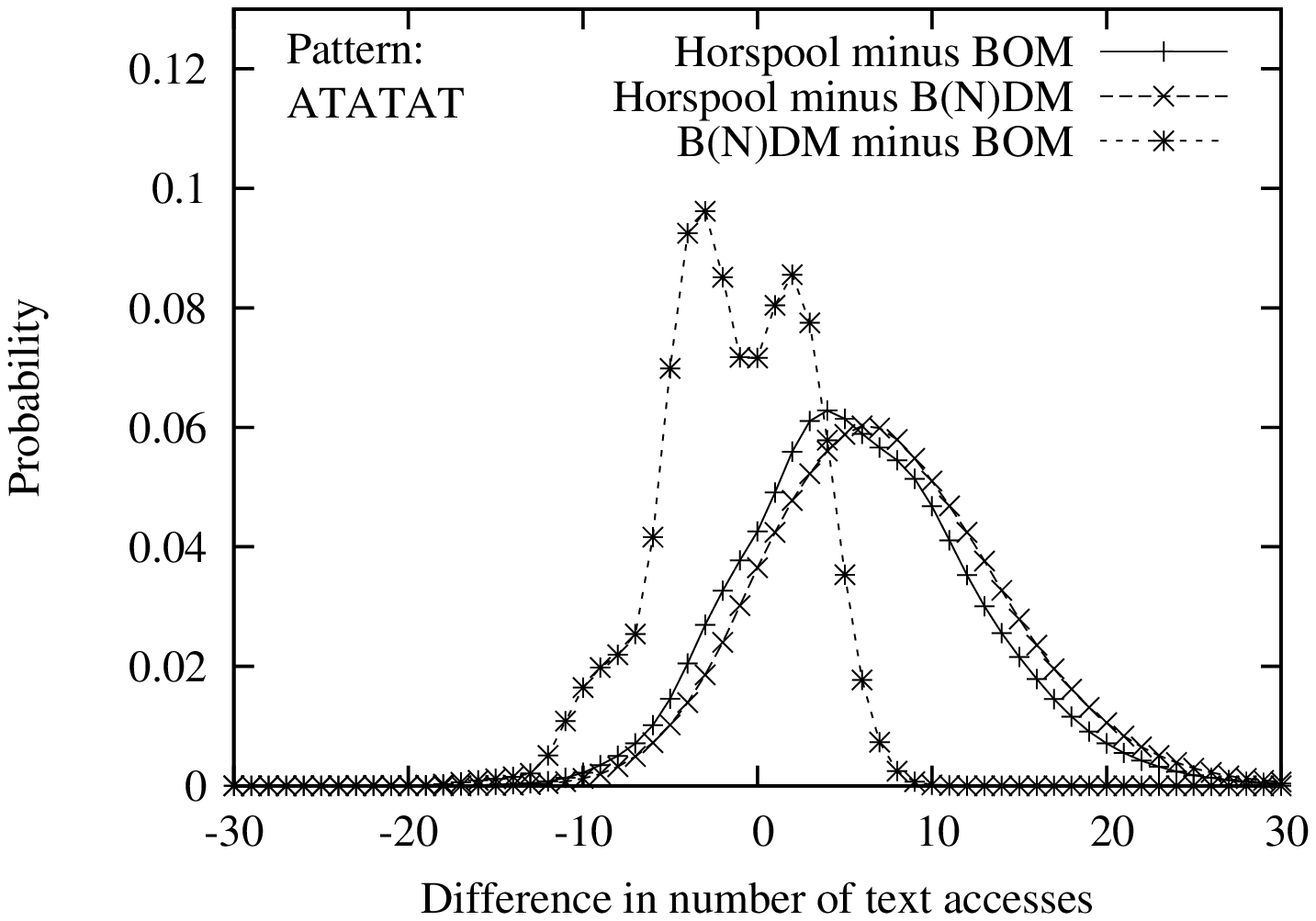}\\
\includegraphics[angle=-90,width=.5\textwidth]{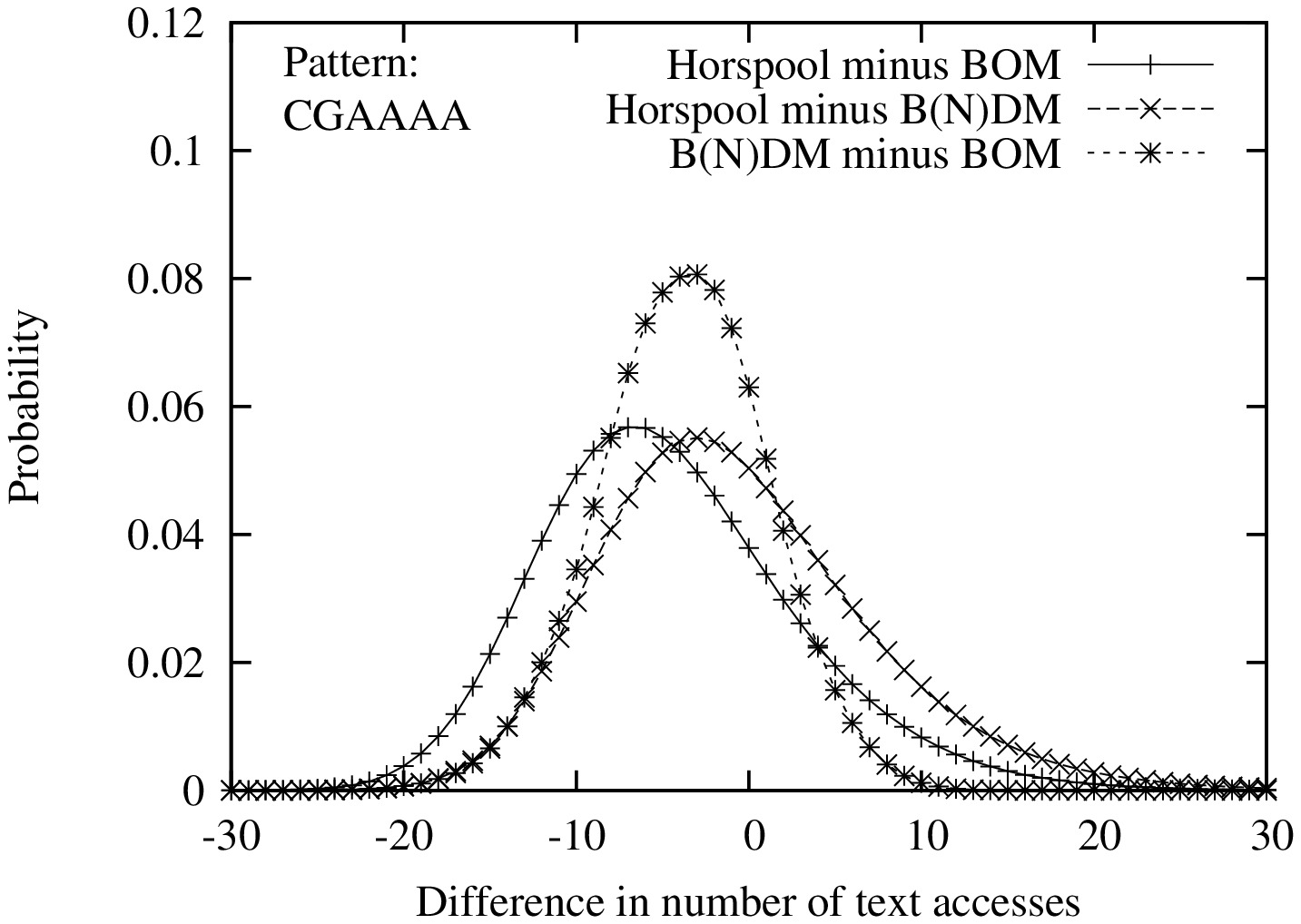}%
\includegraphics[angle=-90,width=.5\textwidth]{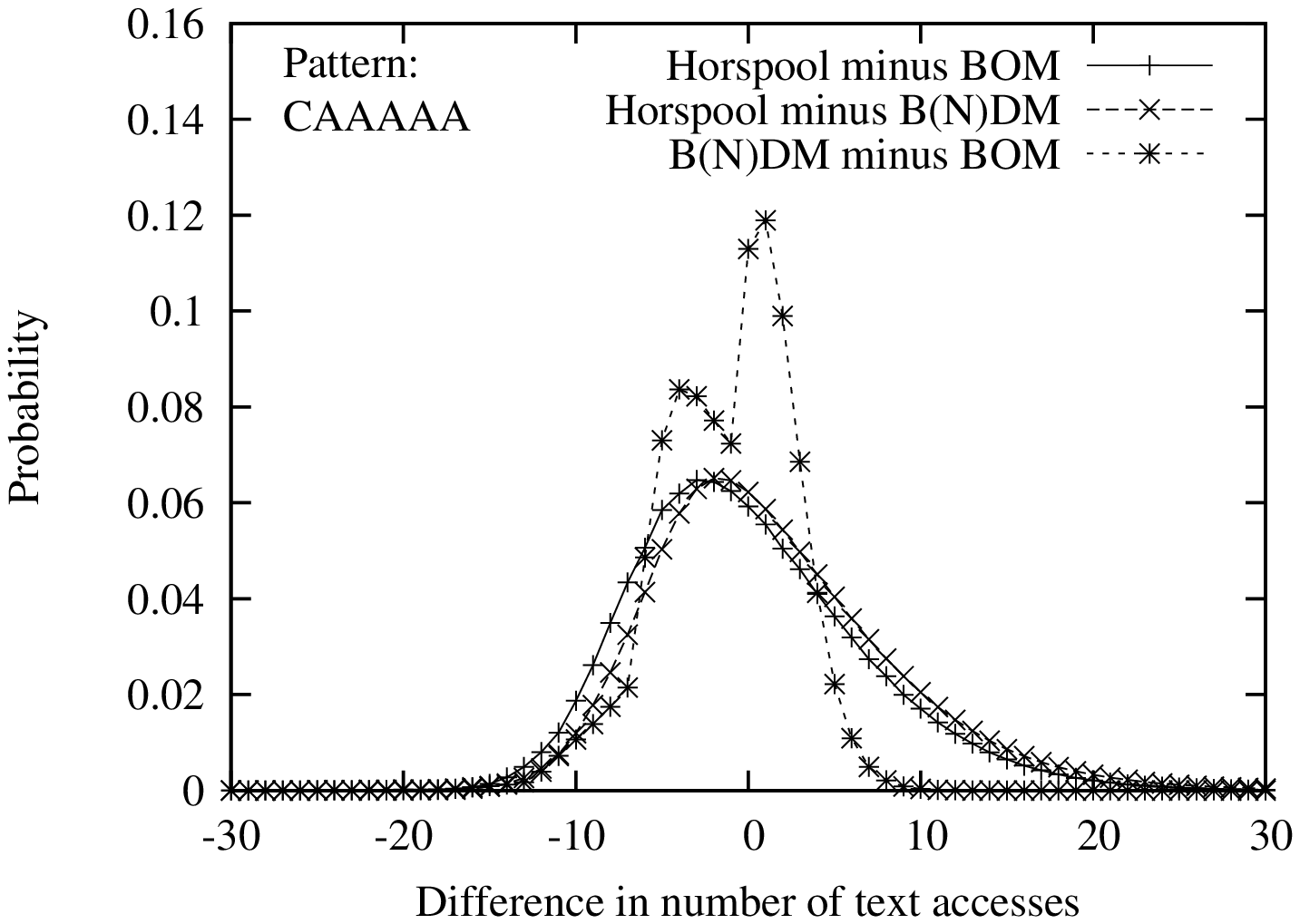}%
\end{center}
\caption{Exact distributions of differences in character access counts for different patterns using a uniform character distribution as text model and random texts of lengths 100.}\label{fig:cost_diffs}
\end{figure}

In Section~\ref{sec:algo}, we considered three practically relevant algorithms, namely Horspool's algorithm, backward oracle matching (BOM), and backward (non)-deterministic DAWG matching (B(N)DM). 
Now, we compare the distributions of running time costs of these algorithms for several patterns. 
Figure~\ref{fig:dists} shows these distributions for the patterns \texttt{ATATAT} and \texttt{ACGTAC} for text lengths 100 and 500 under a uniform i.i.d.\ model on the DNA alphabet $\{\texttt{A,C,G,T}\}$.
For text length 500, the distributions for Horspool and B(N)DM resemble the shape of normal distributions. 
In fact, for Horspool's algorithm it has been proven that the distribution is asymptotically normal \cite{Smythe2001}. For smaller text lengths (e.g.\ 100, as shown in left column of Figure~\ref{fig:dists}), the distributions are less smooth than for longer texts. It is remarkable that for BOM we find zero probabilities with a fixed period. In all examples shown Figure~\ref{fig:dists} this period equals 7.

The probability that one pattern matching algorithm is faster than another depends on the pattern. 
Using the technique introduced in Section~\ref{sec:comparison}, we can quantify the strength of this effect. 
Figure~\ref{fig:cost_diffs} shows distributions of cost \emph{differences} for different patterns and algorithms. 
That means, the probability that the first algorithm is faster is represented by the area under the curve left of zero. 
For the pattern \texttt{CGAAAA}, for example, there is a 55.6\% probability that Horspool's algorithm needs fewer character accesses than B(N)DM in uniform i.i.d.\ texts of length 100, while for \texttt{ACGTAC}, the probability is only 0.18\%. 

Worth noting and perhaps surprising is the fact that there is a non-zero probability of BOM being faster than B(N)DM altough, $\shift{\text{B(N)DM},p}{w} \geq \shift{\text{BOM},p}{w}$ for all window contents~$w$.
The explanation, of course, is that a shorter (say, first) shift for BOM leads to a different window content than for B(N)DM for the second window, which may have a larger shift value.
This effect depends on the pattern: 
For the pattern \texttt{CAAAAA}, there is a 48.2\% probability that BOM performs better than B(N)DM, while it is 6.2\% for \texttt{ACGTAC}, again on texts of length 100.

\begin{table}
\caption{Comparison of DAA sizes for all patterns of length $m$ over $\Sigma=\{\texttt{A},\texttt{C},\texttt{G},\texttt{T}\}$.}\label{tab:automata_sizes}
\begin{center}
\begin{tabular}{ccccc}
\hline
$m$ & States unminimized   & \multicolumn{3}{c}{States minimized (min./avg./max.)} \\
    & $|\Sigma|^m\cdot(m+1)$ &  Horspool & BOM & B(N)DM \\\hline
  2 & 48         & 4 / 4.8 / 5    &  4 / 4.0 / 4   &   4 / 4.8 / 5\\
  3 & 256        & 7 / 8.3 / 9    &  7 / 8.3 / 9   &   7 / 9.6 / 10\\
  4 & 1280       & 11 / 14.3 / 15 & 11 / 15.6 / 18 &  11 / 17.0 / 19\\
  5 & 6144       & 16 / 23.6 / 25 & 16 / 26.5 / 30 &  16 / 27.9 / 31\\
  6 & 28672      & 22 / 37.0 / 39 & 22 / 41.8 / 47 &  22 / 42.8 / 48\\
  7 & 131072     & 29 / 55.2 / 58 & 29 / 62.4 / 70 &  29 / 62.6 / 70\\\hline
\end{tabular}
\end{center}
\end{table}

\begin{figure}[t!]
\begin{center}
\includegraphics[angle=-90,width=\textwidth]{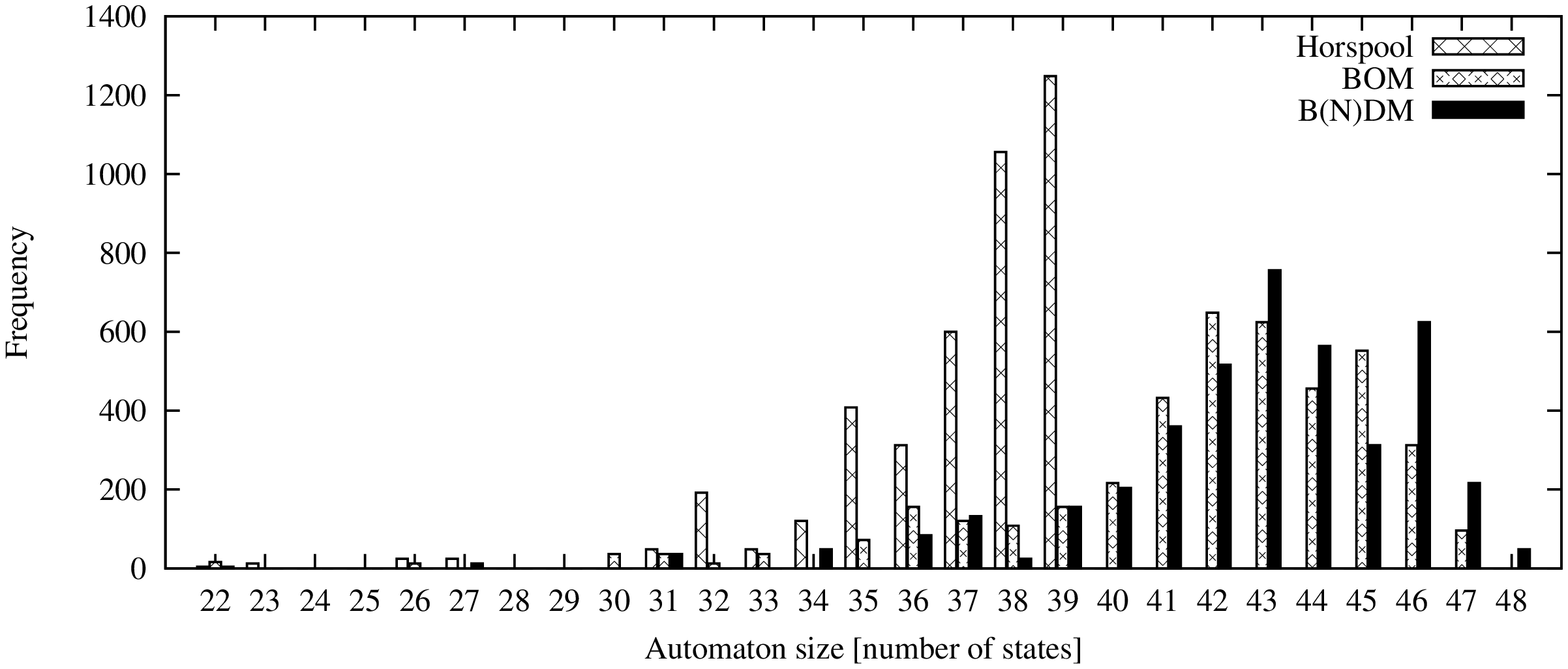}
\end{center}
\caption{Histogram on number of states of minimal DAAs over all patterns of length 6 over $\Sigma=\{\texttt{A},\texttt{C},\texttt{G},\texttt{T}\}$.}\label{fig:automata_sizes}
\end{figure}

To assess the effect of DAA minimization before constructing PAAs,
we constructed minimized DAAs for all 21840~patterns of lengths~2 to~7 over $\Sigma=\{\texttt{A},\texttt{C},\texttt{G},\texttt{T}\}$. 
The minimum, average, and maximum state counts are shown in Table~\ref{tab:automata_sizes}. 
For length 6, Figure~\ref{fig:automata_sizes} contains a detailed histogram. 
These statistics show that construction and minimization as given in this article lead to smaller automata (and thus better runtimes) than the constructions given in the conference version of this article~\cite{Marschall2010}.
It may be conjectured that that minimal state spaces grows only polynomial with~$m$ for all of these algorithms, as has been previously proven for the Horspool algorithm~\cite{Tsai2006}.

A JAVA implementation is available at \url{http://www.rahmannlab.de/software}.
All algorithms were run on an Intel Core 2 Quad CPU at 2.66GHz.
Computing the distributions shown in Figure~\ref{fig:dists} took 0.3 to 0.6 seconds for each distribution. 
Distributions of differences as in Figure~\ref{fig:cost_diffs} were computed in 14 to 36 seconds.


\section{Discussion}
\label{sec:discussion}

Using PAAs, we have shown how the exact distribution of the number of character accesses for window-based pattern matching algorithms can be computed. 
The framework is general enough to admit i.i.d.\ text models, Markovian text models of arbitrary order, and character-emitting hidden Markov models. 
The given construction results in an asymptotic runtime of $O(\totalsteps^2\cdot m\cdot |\stateset^\daa|\cdot|\mathcal{C}|^2\cdot |\Sigma|)$. 
The number of DAA states $\stateset^\daa$ is $O(m\cdot\Sigma^m)$, but it can be considerably reduced by DAA minimization. 
The resulting PAA is smaller and therefore computing the cost distribution is much faster. 
If the pattern length~$m$ is large, however, construction and minimization of the DAA itself pose a significant burden. 
It remains open if there exists an algorithm to construct the minimal automaton directly in general, i.e.\ using only $O(|\stateset^\daa_{\min}|)$ time. For Horspool's algorithm, an automaton with $O(m^2)$ states can be constructed by using the ideas from~\cite{Tsai2006}. The construction, however, cannot be easily transferred to other algorithms and we are not aware of any similar results for BOM or B(N)DM.

In this work, we considered the most practically relevant algorithms.
Exemplarily studying the cost distribution for some patterns showed that the algorithms performance is indeed non-negligibly influenced by the choice of the pattern.
Especially the behavior of BOM deserves further attention: Its distribution of text character accesses features periodic zero probabilities, and unexpectedly, it may need fewer text accesses than B(N)DM on some patterns, although BOM's shift values are never better than B(N)DM's.

We focused on algorithms for single patterns, but the presented techniques also apply to algorithms to search for multiple patterns like the Wu-Manber algorithm~\cite{Wu1994} or  set backward oracle matching and multiple BNDM as given in~\cite{Navarro2002}. A comparison of the resulting distributions could yield new insights into these algorithms as well.

Other metrics than text character accesses might be of interest and could be easily substituted; for example, just counting the number of windows by defining $\cost{p}{w}=1$ for all $w\in\Sigma^m$.

The given constructions allow us to analyze an algorithm's performance for each pattern individually. 
While this is desirable for detailed analysis, the cost distribution resulting from randomly choosing text \emph{and} pattern would also be of interest.


\bibliographystyle{abbrv}
\bibliography{analysis}

\begin{thebibliography}{10}

\bibitem{Allauzen2001}
C.~Allauzen, M.~Crochemore, and M.~Raffinot.
\newblock Efficient experimental string matching by weak factor recognition.
\newblock In G.~Goos, J.~Hartmanis, and J.~van Leeuwen, editors, {\em
  Proceedings of the 12th Annual Symposium on Combinatorial Pattern Matching
  (CPM)}, volume 2089 of {\em LNCS}, pages 51--72, 2001.

\bibitem{Baeza-Yates1990}
R.~A. Baeza-Yates, G.~H. Gonnet, and M.~R\'{e}gnier.
\newblock Analysis of {Boyer-Moore}-type string searching algorithms.
\newblock In {\em SODA '90: Proceedings of the first annual ACM-SIAM symposium
  on Discrete algorithms}, pages 328--343. SIAM, 1990.

\bibitem{BaezaYatesRegnier1992}
R.~A. Baeza-Yates and M.~R{\'e}gnier.
\newblock Average running time of the {B}oyer-{M}oore-{H}orspool algorithm.
\newblock {\em Theor. Comput. Sci.}, 92(1):19--31, 1992.

\bibitem{Boyer1977}
R.~S. Boyer and J.~S. Moore.
\newblock A fast string searching algorithm.
\newblock {\em Communications of the ACM}, 20(10):762--772, 1977.

\bibitem{Crochemore1994}
M.~Crochemore, A.~Czumaj, L.~Gasieniec, S.~Jarominek, T.~Lecroq, W.~Plandowski,
  and W.~Rytter.
\newblock Speeding up two string-matching algorithms.
\newblock {\em Algorithmica}, 12(4--5):247--267, 1994.

\bibitem{HermsRahmann2008ComputingSeedSensitivity}
I.~Herms and S.~Rahmann.
\newblock Computing alignment seed sensitivity with probabilistic arithmetic
  automata.
\newblock In K.~Crandall and J.~Lagergren, editors, {\em Algorithms in
  Bioinformatics (WABI)}, volume 5251 of {\em LNCS}, pages 318--329. Springer,
  2008.

\bibitem{Hopcroft1971}
J.~Hopcroft.
\newblock An $n \log n$ algorithm for minimizing the states in a finite
  automaton.
\newblock In Z.~Kohavi and A.~Paz, editors, {\em The theory of machines and
  computations}, pages 189--196. Academic Press, New York, 1971.

\bibitem{Horspool1980}
R.~N. Horspool.
\newblock Practical fast searching in strings.
\newblock {\em Software-Practice and Experience}, 10:501--506, 1980.

\bibitem{KnuthMorrisPratt1977}
D.~E. Knuth, J.~Morris, and V.~R. Pratt.
\newblock Fast pattern matching in strings.
\newblock {\em {SIAM} Journal on Computing}, 6(2):323--350, 1977.

\bibitem{Knuutila2001}
T.~Knuutila.
\newblock Re-describing an algorithm by {H}opcroft.
\newblock {\em Theoretical Computer Science}, 250(1-2):333--363, January 2001.

\bibitem{Kucherov2006}
G.~Kucherov, L.~No\'{e}, and M.~Roytberg.
\newblock A unifying framework for seed sensitivity and its application to
  subset seeds.
\newblock {\em Journal of Bioinformatics and Computational Biology},
  4(2):553--569, 2006.

\bibitem{Mahmoud1997}
H.~M. Mahmoud, R.~T. Smythe, and M.~R\'{e}gnier.
\newblock Analysis of {Boyer-Moore-Horspool} string-matching heuristic.
\newblock {\em Random Structures and Algorithms}, 10(1-2):169--186, 1997.

\bibitem{Marschall2008}
T.~Marschall and S.~Rahmann.
\newblock Probabilistic arithmetic automata and their application to pattern
  matching statistics.
\newblock In P.~Ferragina and G.~M. Landau, editors, {\em Combinatorial Pattern
  Matching (CPM)}, volume 5029 of {\em LNCS}, pages 95--106. Springer, 2008.

\bibitem{MarschallRahmann2009EfficientExactMotifDiscovery}
T.~Marschall and S.~Rahmann.
\newblock Efficient exact motif discovery.
\newblock {\em Bioinformatics}, 25(12):i356--i364, 2009.

\bibitem{Marschall2010}
T.~Marschall and S.~Rahmann.
\newblock Exact analysis of {H}orspool's and {S}unday's pattern matching
  algorithms with probabilistic arithmetic automata.
\newblock In A.-H. Dediu, H.~Fernau, and C.~Mart\'{i}n-Vide, editors, {\em
  Proceedings of the 4th International Conference on Language and Automata
  Theory and Applications (LATA)}, volume 6031 of {\em LNCS}, pages 439--450,
  2010.

\bibitem{Navarro2002}
G.~Navarro and M.~Raffinot.
\newblock {\em Flexible Pattern Matching in Strings}.
\newblock {Cambridge University Press}, 2002.

\bibitem{Schulz2008}
M.~Schulz, D.~Weese, T.~Rausch, A.~D\"{o}ring, K.~Reinert, and M.~Vingron.
\newblock Fast and adaptive variable order markov chain construction.
\newblock In K.~A. Crandall and J.~Lagergren, editors, {\em Algorithms in
  Bioinformatics (WABI'08)}, volume 5251 of {\em LNCS}, pages 306--317.
  Springer, 2008.

\bibitem{Smythe2001}
R.~T. Smythe.
\newblock The {B}oyer-{M}oore-{H}orspool heuristic with {M}arkovian input.
\newblock {\em Random Structures and Algorithms}, 18(2):153--163, 2001.

\bibitem{Sunday1990}
D.~M. Sunday.
\newblock A very fast substring search algorithm.
\newblock {\em Communications of the {ACM}}, 33(8):132--142, 1990.

\bibitem{Tsai2006}
T.~Tsai.
\newblock Average case analysis of the {Boyer-Moore} algorithm.
\newblock {\em Random Structures and Algorithms}, 28(4):481--498, 2006.

\bibitem{Wu1994}
S.~Wu and U.~Manber.
\newblock A fast algorithm for multi-pattern searching.
\newblock Technical report, University of Arizona, Tucson, AZ, 1994.

\end{thebibliography}

\end{document}